\documentclass[12pt]{iopart}
\usepackage{amssymb,latexsym,amsthm}
\usepackage{cite}

\theoremstyle{remark}
\newtheorem{remark}{Remark}

\theoremstyle{plain}
\newtheorem{theorem}{Theorem}
\newtheorem{proposition}{Proposition}

\begin{document}
\title[Semiclassical evolution on the torus beyond the Ehrenfest time]{Semiclassical evolution of quantum wave packets on the torus beyond the Ehrenfest time in terms of Husimi distributions}
\author{A S Trushechkin}
\address{Steklov Mathematical Institute of Russian Academy of Sciences, 119991 Moscow, Russia}
\ead{trushechkin@mi.ras.ru}
\begin{abstract}
The semiclassical long-time limit of free evolution of quantum wave packets on the torus is under consideration. Despite of simplicity of this system, there are still open questions concerning the detailed description of the evolution on time scales beyond the Ehrenfest time. One of the approaches is based on the limiting Wigner or Husimi distributions of time-evolved wave packets as the Planck constant tends to zero and time tends to infinity. We derive explicit expressions for semiclassical measures corresponding to all time scales and the corresponding stages of evolution: classical-like motion, spreading of the wave packet, and its revivals. Also we discuss limitations of the approach based on semiclassical measures and suggest its generalization.
\end{abstract}
\maketitle

\section{Introduction}

The dynamics of a localized quantum wave packet in a finite region or on a compact manifold on short time scales is well-known to be described the classical motion of its center and gradual spreading. The characteristic time scale when this description breaks down is called the Ehrenfest time. The Ehrenfest time is estimated as $O(\ln\hbar^{-1})$ (where $\hbar$ is the Planck constant), although it may be larger for integrable systems (see rigorous results in \cite{CR,Bambusi,Hage1,Hage2,Bouzounia,Schubert}). The description of semiclassical evolution of quantum wave packets at the Ehrenfest time and beyond it attracts much attention \cite{BerrySpin,Schubert,SemiclWaveRev,WangHeller,Schubert-spread,MaciaRiemann, MaciaTorus, AnaMaciaTorus, AnaMaciaView, Ana14}. Mathematically, can be formulated as the simultaneous limit when the Planck constant goes to zero and time goes to infinity. We will refer to this type of limits as semiclassical long-time limits \cite{BerrySpin}.

One of directions of researches is related to the so called semiclassical measures, i.e., semiclassical limit of Wigner measures \cite{Gerard, Mark, Carles, Numeric}. In general, the description of semiclassical dynamics in the Wigner--Weyl representation is quite popuar \cite{Almedia2013,Almedia2016,Gosson}. In \cite{MaciaRiemann, MaciaTorus, AnaMaciaTorus, AnaMaciaView, Ana14}, a number of properties of semiclassical measures related to times beyond the Ehrenfest time have been obtained. However, explicit calculation of semiclassical measures even for simplest cases presents certain difficulties. In particular, in \cite{AnaMaciaView} this problem is characterised as 'notoriously difficult'. 

The result of the present work is explicit calculation of semiclassical measures related to the free dynamics of quantum wave packets on the flat torus $\mathbb T^d=\mathbb R^d/(2\pi\mathbb Z^d)$. We generalize the results of \cite{VolTrush} where only Gaussian wave packets are considered.

Also our results provide further insights about the limitations of the approach to long-time semiclassical dynamics based on semiclassical measures (reported in \cite{Carles}) and propose its generalizations.

The usual way to deal with quantum dynamics in the semiclassical approximation is to reduce it to corresponding problems in classical dynamics. Here we adopt an alternative approach based on direct summation of series of eigenvectors for time-evolved wavepackets. An application of this approach to the Jaynes--Cummings model is given in \cite{Kara}.

The following text is organised as follows. Preliminary facts about Wigner and Husimi measures, semiclassical measures, and coherent states are given in~\ref{SecPrelim}. Also we prove some intermediate results there. The main results (Theorems \ref{ThMain}--\ref{ThTime}) are stated and proved in~\ref{SecMain}. Theorem~\ref{ThMain}  is the main one, while Theorems~\ref{ThMu} and~\ref{ThTime} are corollaries of Theorem~\ref{ThMain} and intermediate formulas obtained in its proof. In Sec.~\ref{SecDiscus} we discuss the results.

\section{Preliminaries}\label{SecPrelim}

\subsection{Schr\"odinger equation on the flat torus}
Consider the Schr\"odinger equation on the flat torus $\mathbb T^d=\mathbb R^d/(2\pi\mathbb Z^d)$:

\begin{equation}\label{EqSchr}
i\hbar\frac{\partial \psi_t}{\partial t}=-\hbar^2\Delta\psi_t,
\end{equation}
where $\psi_t=\psi_t(x)$, $t\in\mathbb R$, $x\in\mathbb T^d$, $\Delta$ is the Laplace operator over the spatial variables $x$, $\hbar>0$ is the Planck constant. The solution of the Cauchy problem with some initial function 
\begin{equation*}
\psi_0(x)=\frac1{(2\pi)^{\frac d2}}\sum_{k\in\mathbb Z^d}c^{(0)}_k\exp(ikx)\in L^2(\mathbb T^d)
\end{equation*}
can be formally represented as an action of a unitary operator in $L^2(\mathbb T^d)$:
\begin{equation}\label{EqEvol}
\psi_t(x)=\exp(-i\hbar t\Delta)\psi_0(x)=\frac1{(2\pi)^{\frac d2}}\sum_{k\in\mathbb Z^d}c^{(0)}_k\exp(ikx-i\hbar t k^2).
\end{equation}

Formula (\ref{EqEvol}) directly implies that every solution of (\ref{EqSchr}) is periodic with the period 
\begin{equation}\label{EqT}
T_\hbar=\frac{2\pi}\hbar,
\end{equation}
i.e. $\psi_{t+T_\hbar}=\psi_t$. The time $T_\hbar$ is called the revival time. This periodicity is caused by interference and has purely wave nature. As $\hbar\to0$, the revival time tends to infinity.

\subsection{Semiclassical measures}
We will identify functions on $\mathbb T^d$ with $(2\pi\mathbb Z^d)$-periodic functions on  $\mathbb R^d$. Then, the Wigner distribution on the phase space $\Omega=\mathbb T^d\times\mathbb R^d$ for an arbitrary function $\psi\in L^2(\mathbb T^d)$ is defined as \cite{QMPS,Hillery}

\begin{eqnarray}
W_\psi(q,p)&=\frac1{(\pi\hbar)^d}\int_{\mathbb R^d}\overline{\psi(q+x)}\psi(q-x)\exp\left(\frac{2ipx}\hbar\right)dx\nonumber\\&=\frac1{(2\pi)^d}\sum_{j,k\in\mathbb Z^d}\overline{c_j}c_k\exp[i(k-j)q]\,\delta\left(p-\frac\hbar2(k+j)\right),\label{EqWigner}
\end{eqnarray}
where $c_k$ are the Fourier coefficients of $\psi(x)=(2\pi)^{-d/2}\sum_kc_k\exp(ikx)$, $\delta(\cdot)$ is the Dirac delta function, and $(q,p)\in\mathbb T^d\times\mathbb R^d$. An important property of the Wigner distribution is that its marginal distributions over $q$ and $p$ coincide with the corresponding quantum-mechanical distributions:
\numparts
\begin{eqnarray}\label{EqWignerProp}
\int_{\mathbb R^d} W_\psi(q,p)\,dp&=|\psi(q)|^2,\\
\int_{\mathbb T^d} W_\psi(q,p)\,dq&= \sum_{k\in\mathbb Z^d}|c_k|^2\delta(p-\hbar k).
\label{EqWignerProp2}
\end{eqnarray}
\endnumparts
However, the Wigner distribution is generally non-positive. By this reason,  sometimes it is called quasiprobability distribution.

Consider a family of functions $\{\psi_\hbar\}$ depending on $\hbar$ and the corresponding Wigner distributions $W_{\psi_\hbar}$. For shortness, we will write $W_\hbar$ if $\psi$ is fixed. If there exists a measure $\mu$ on $\Omega$ such that there exists a limit
\begin{equation}\label{EqSemiM}
\lim_{\hbar\to0}\int_\Omega W_{\hbar}(q,p)a(q,p)\,dqdp=\int_\Omega a(q,p)\mu(dqdp)
\end{equation}
for all functions $a\in C^\infty_0(\Omega)$ (infinitely differentiable functions with compact supports), then the measure $\mu$ is called the semiclassical measure \cite{Zworski,MaciaRiemann,MaciaTorus}. 

\begin{remark}
It is always possible to choose a proper sequence $\{\psi_{\hbar_n}\}$ such that limit (\ref{EqSemiM}) exists for this sequence.
\end{remark}

\begin{remark}
The Planck constant $\hbar$ is a fundamental physical constant with  dimensions of action. So, rigorously speaking, it cannot tend to zero. This formal mathematical limit means that the Planck constant is much smaller than another quantity with  dimensions of action arising in a concrete problem. In Section~\ref{SecPhysSmall} we will describe such kind of conditions for our case. 
\end{remark}

We will adopt another, equivalent, approach to the semiclassical measures, which is based not on the Wigner distribution, but on the Husimi distribution. For this purpose, we need to define coherent states on the torus.

\subsection{Coherent states}

Consider a smooth rapidly decreasing function $\varphi(x)$, $x\in \mathbb R^d$, with unit $L^2(\mathbb R^d)$-norm and a family of functions from $L^2(\mathbb R^d)$ of the form
\begin{equation}\label{EqCoherRd}
\eta^{(\hbar)}_{qp}(x)=\frac1{\sqrt{\alpha_\hbar^d}}\varphi\left(\frac{x-q}{\alpha_\hbar}\right)
\exp\left\lbrace\frac{ip(x-q)}\hbar\right\rbrace,
\end{equation}
where $(q,p)\in\mathbb R^{2d}$ and $\alpha_\hbar>0$ is a constant depending on $\hbar$ such that $\alpha_\hbar\to0$ and $\hbar/\alpha_\hbar\to0$ as $\hbar\to0$ (e.g., $\alpha_\hbar=\sqrt\hbar$). These functions satisfy the general definition of coherent states on $L^2(\mathbb R^d)$ given in \cite{Klauder}: this family of functions continuously depends in its parameters $(q,p)$ and constitutes a continuous resolution of identity:
\begin{equation}\label{EqUnityRd}
\frac1{(2\pi\hbar)^d}\int_{\mathbb R^2} P[\eta^{(\hbar)}_{qp}]\,dqdp=1.
\end{equation}
Here $P[\psi]$ is an operator acting on an arbitrary vector  $\chi$ as $P[\psi]\chi=(\psi,\chi)\psi$ ($P[\psi]=|\psi\rangle\langle\psi|$ in the Dirac notations; it is a projector whenever $\psi$ is a unit vector); $(\cdot,\cdot)$ is a scalar product (with linearity in the second argument). Equality (\ref{EqUnityRd}) is understood in the weak sense: for all $\psi,\chi\in L^2(\mathbb R^d)$ we have
\begin{equation*}
\frac1{(2\pi\hbar)^d)}\int_{\mathbb R^2} (\psi,\eta^{(\hbar)}_{qp})(\eta^{(\hbar)}_{qp},\chi)\,dqdp=(\psi,\chi).
\end{equation*}
Usually coherent states are required to correspond to classical particles in some way. Let us proof the following known property (we will use it subsequently).

\begin{proposition}\label{PropDeltaLine}
The semiclassical measure of the family of functions $\eta^{(\hbar)}_{q_0p_0}$ is  the Dirac measure at $(q_0,p_0)$.
\end{proposition}

\begin{proof}
Denote $W_{q_0,p_0}^{(\hbar)}$ the Wigner distribution corresponding to the wave function $\upsilon^{(\hbar)}_{q_0p_0}$. By (\ref{EqWigner}),
\begin{eqnarray}\fl
W_{q_0,p_0}^{(\hbar)}(q,p)=\frac1{(\pi\hbar)^d}\int_{\mathbb R^d}
\overline{\eta^{(\hbar)}_{q_0p_0}(q+x)}\eta^{(\hbar)}_{q_0p_0}(q-x)\exp\left(\frac{2ipx}\hbar\right)dx\nonumber\\=
\frac1{(\pi\hbar\alpha_\hbar)^d}\int_{\mathbb R^d}
\overline{\varphi^{(\hbar)}\left(\frac{q-q_0+x}{\alpha_\hbar}\right)}
\varphi^{(\hbar)}\left(\frac{q-q_0-x}{\alpha_\hbar}\right)\nonumber\\\quad\times
\exp\left\lbrace\frac{2ipx+ip_0(q-q_0-x)-ip_0(q-q_0+x)}\hbar\right\rbrace dx\nonumber\\=
\frac1{(\pi\hbar)^d}\int_{\mathbb R^d}
\overline{\varphi^{(\hbar)}\left(\frac{q-q_0}{\alpha_\hbar}+x\right)}
\varphi^{(\hbar)}\left(\frac{q-q_0}{\alpha_\hbar}-x\right)\nonumber\\\quad\times
\exp\left\lbrace\frac{2i(p-p_0)\alpha_\hbar x}\hbar\right\rbrace dx\nonumber\\=
\left(\frac{\alpha_\hbar}\hbar\right)^d\left(\frac1{\alpha_\hbar}\right)^d 
f\left(\frac{q-q_0}{\alpha_\hbar},\frac{\alpha_\hbar}\hbar(p-p_0)\right),\label{EqWdelta}
\end{eqnarray}
where
\begin{equation*}
f(q,p)=\frac1{\pi^d}\int_{\mathbb R^d}\overline{\varphi(q+x)}\varphi(q-x)\exp(2ipx)dx.
\end{equation*}
Since $\int_{\mathbb R^d}f(q,p)dqdp=1$, expression (\ref{EqWdelta}) implies
\begin{equation*}
\lim_{\hbar\to0}W_{q_0,p_0}^{(\hbar)}(q,p)=\delta(q-q_0)\delta(p-p_0).
\end{equation*}
\end{proof}

The Dirac measure at $(q_0,p_0)$  corresponds to a classical particle in this phase point  (for the configuration space $\mathbb R^d$; we will return to the case of torus a bit later). Time evolution of this semiclassical measure on short times also can be shown to correspond to the classical phase trajectory.

A particular case are Gaussian coherent states, which correspond to the following choice of  the function $\varphi$:
\begin{equation}\label{EqGauss}
\varphi(x)=\frac{1}{(2\pi)^{\frac d4}}\exp\left(-\frac{x^2}4\right).
\end{equation}
In this case, $\alpha_\hbar$ and $\hbar/(2\alpha_\hbar)$ are the standard deviations of the position and the momentum respectively. Their product gives $\hbar/2$, so, the Gaussian coherent states minimize the uncertainty relations. 

Functions of form (\ref{EqCoherRd}) are also referred to as quantum wave packets because they are superpositions of monochromatic waves $\exp(ipx)$ and they are localized in both position and momentum spaces.

On the base of coherent states (\ref{EqCoherRd}) on $\mathbb R^d$, coherent states on the torus $\mathbb T^d$ can be constructed as follows \cite{DG,KR96,Gonzalez,KR07,KR08}:
\begin{equation}\label{EqCoher}
\upsilon^{(\hbar)}_{qp}(x)=\sum_{n\in\mathbb Z^d}\eta^{(\hbar)}_{qp}(x-2\pi n),
\end{equation}
where $(q,p)\in\Omega$. They also constitute a continuous resolution on identity:
\begin{equation}\label{EqUnity}
\frac1{{(2\pi\hbar)^d}}\int_\Omega P[\upsilon^{(\hbar)}_{qp}]\,dqdp=1.
\end{equation}
Let us note that the functions $\upsilon^{(\hbar)}_{qp}$ as elements of $L^2(\mathbb T^d)$ are not normalized to unity. However, their norms tend to unity as $\hbar\to0$. Indeed,
\begin{eqnarray*}
\|\upsilon^{(\hbar)}_{qp}\|^2&=\int_{\mathbb T^d}\overline{\upsilon^{(\hbar)}_{qp}(x)}\upsilon^{(\hbar)}_{qp}(x)\,dx=
\sum_{n\in\mathbb Z^d}\int_{\mathbb T^d}
\overline{\eta^{(\hbar)}_{qp}(x-2\pi n)}\upsilon^{(\hbar)}_{qp}(x)\,dx\\&=
\int_{\mathbb R^d}
\overline{\eta^{(\hbar)}_{qp}(x)}\upsilon^{(\hbar)}_{qp}(x)\,dx=
\sum_{m\in\mathbb Z^d}\int_{\mathbb R^d}
\overline{\eta^{(\hbar)}_{qp}(x)}\eta^{(\hbar)}_{qp}(x-2\pi m)\,dx\\&=
1+\sum_{m\in\mathbb Z^d\backslash\{0\}}\int_{\mathbb R^d}
\overline{\eta^{(\hbar)}_{qp}(x)}\eta^{(\hbar)}_{qp}(x-2\pi m)\,dx.
\end{eqnarray*}
We have used that the functions $\eta^{(\hbar)}_{qp}(x)\in L^2(\mathbb R^d)$ have the unit norm. Since the function $\varphi$ rapidly decreases, the last expression tends to unity.

From now $W_{q_0,p_0}^{(\hbar)}$ will denote the Wigner distribution corresponding to the wave function $\upsilon^{(\hbar)}_{q_0p_0}$. We will use the following property of the distribution $W_{q_0,p_0}^{(\hbar)}$.

\begin{proposition}\label{PropUni}
\numparts
\begin{eqnarray}\label{EqWquni}
&W_{q_0,p_0}^{(\hbar)}(q,p)&=W_{q_0+\Delta q,p_0}^{(\hbar)}(q+\Delta q,p),\\
&W_{q_0,p_0}^{(\hbar)}(q,p)&=W_{q_0,p_0+\Delta p}^{(\hbar)}(q,p+\Delta p)+o(1),\quad \hbar\to0\label{EqWpuni}
\end{eqnarray}
\endnumparts
\end{proposition}

\begin{proof}
The first equality is obvious from the definitions of the Wigner distribution and the functions $\eta$ and $\upsilon$. For the proof of the second inequality, we firstly note that
$$\eta_{q_0,p_0+\Delta p_0}^{(\hbar)}(x)=\eta_{q_0,p_0}^{(\hbar)}(x)
\exp\left(-\frac{i\Delta p(x-q)}\hbar\right),$$
\begin{equation}\label{EqUspEta}
\upsilon_{q_0,p_0}^{(\hbar)}(x)=\eta_{q_0,p_0}^{(\hbar)}(x-2\pi n_{x-q_0})+
o(1),\quad\hbar\to0,\end{equation}
where $n_y$ denotes the integer with the property $y-2\pi n_y\in[-\pi,\pi)^d$ for an arbitrary real $y$. Hence,
$$\upsilon_{q_0,p_0+\Delta p_0}^{(\hbar)}(x)=\upsilon_{q_0,p_0}^{(\hbar)}(x)
\exp\left(-\frac{i\Delta p(x-q-2\pi n_{x-q_0})}\hbar\right)+
o(1),\quad\hbar\to0.$$
Then,
\begin{eqnarray*}\fl
W_{q_0,p_0+\Delta p}^{(\hbar)}(q,p+\Delta p)\!=
\frac1{(\pi\hbar)^d}\int_{\mathbb R^d}
\overline{\upsilon^{(\hbar)}_{q_0,p_0+\Delta p}(q+x)}
\upsilon^{(\hbar)}_{q_0,p_0+\Delta p}(q-x)
\exp\left(\frac{2i(p+\Delta p)x}\hbar\right)\!dx\\\fl
=\frac1{(\pi\hbar)^d}\int_{\mathbb R^d}
\overline{\upsilon^{(\hbar)}_{q_0,p_0}(q+x)}
\upsilon^{(\hbar)}_{q_0,p_0}(q-x)
\exp\left(\frac{2ipx}\hbar\right)
\exp\left(\frac{2i\pi \Delta p}\hbar(n_{q-q_0+x}-n_{q-q_0-x})\right)dx\\\fl
+o(1).
\end{eqnarray*}
Due to the highly oscillating term $\exp(2ipx/\hbar)$ (where $p$ is a variable of integration with a test function), the integration over $x$ is actually performed in an infinitesimal neighbourhood of zero. Hence, $n_{q-q_0+x}=n_{q-q_0-x}$ unless $q-q_0=\pi k$ for some $k$, and (\ref{EqWpuni}) is proved. If $q-q_0=\pi k$, then both  $W_{q_0,p_0}^{(\hbar)}(q,p)$ and $W_{q_0,p_0+\Delta p}^{(\hbar)}(q,p+\Delta p)$ are infinitesimal and (\ref{EqWpuni}) is obviously true.
\end{proof}

\begin{proposition}\label{PropDelta}
The semiclassical measure of the family of functions $\upsilon^{(\hbar)}_{q_0p_0}$ is sum of the Dirac measures at the points $(q_0+2\pi n,p_0)$, $n\in\mathbb Z^d$.
\end{proposition}
\begin{proof}
Using (\ref{EqUspEta}),
\begin{eqnarray*}\fl
W_{q_0,p_0}^{(\hbar)}(q,p)=\frac1{(\pi\hbar)^d}\int_{\mathbb R^d}
\overline{\upsilon^{(\hbar)}_{q_0p_0}(q+x)}\upsilon^{(\hbar)}_{q_0p_0}(q-x)\exp\left(\frac{2ipx}\hbar\right)dx\\\fl=
\frac1{(\pi\hbar)^d}\int_{\mathbb R^d}
\overline{\eta^{(\hbar)}_{q_0p_0}(q+x-2\pi n_{q-q_0+x})}
\eta^{(\hbar)}_{q_0p_0}(q-x-2\pi n_{q-q_0-x})\exp\left(\frac{2ipx}\hbar\right)dx+o(1).\end{eqnarray*}
Using the same reasonings as in the proof of Proposition~\ref{PropUni}, we can put 
$$n_{q-q_0+x}=n_{q-q_0-x}=n_{q-q_0}.$$
Then, due to Proposition~\ref{PropDeltaLine},
$$\lim_{\hbar\to0}W_{q_0,p_0}^{(\hbar)}(q,p)=
\delta(q-q_0-2\pi n_{q-q_0})\delta(p-p_0),$$
or
$$\lim_{\hbar\to0}W_{q_0,p_0}^{(\hbar)}(q,p)=
\sum_{n\in\mathbb Z^d}\delta(q-q_0-2\pi n)\delta(p-p_0).$$

\end{proof}

We mentioned that the Gaussian coherent states on $\mathbb{R}^d$ minimize the uncertainty relations. The uncertainty relations require modifications for compact manifolds (e.g., torus) and bounded domains (e.g., infinite square well). There are several analogues of uncertainty relations for these cases. The  Gaussian coherent states on the torus minimize a variant of the uncertainty relations for the torus \cite{KR96,KR07,KRUncert}. Also some estimates of the standard deviations of position and momentum have been obtained in \cite{VolTrush-Trudy}.

\subsection{Husimi distribution}

For an arbitrary function $\psi\in L^2(\mathbb T^d)$ with unit norm, let us define the probability distribution on the phase space $\Omega$ as
\begin{equation*}
H_\psi(q,p)=\frac1{(2\pi\hbar)^d}|(\upsilon_{qp},\psi)|^2.
\end{equation*}
It is called the Husimi distribution (or the Husimi function) associated to $\psi$ \cite{QMPS, Hillery}
In contrast to the Wigner distribution, the Husimi distribution is positive by construction. The normalization condition
\begin{equation*}
\int_\Omega H_\psi(q,p)\,dqdp=1
\end{equation*}
is satisfied due to resolution of identity  (\ref{EqUnity}) by coherent states. However, marginal position and momentum distributions does not coincide with the original quantum-mechanical distributions (in contrast to (\ref{EqWignerProp}), (\ref{EqWignerProp2})). However, the Husimi distribution has a direct physical meaning -- see Remark~\ref{RemHusimi} below.

The tomographic represetation of quantum mechanics is proposed \cite{Manko,MankoTMF} to overcome the drawbacks of different phase space distributions corresponding to quantum states \cite{QMPS}. In the tomographic representation, not a single, but a family of probability distributions is assigned to a quantum state. Relations between the tomographic representation and the Husimi distribution is considered in \cite{MankoHusimi}.

The notion of semiclassical measures can be equivalently reformulated in terms of the Husimi distribution. Again, consider a family of functions $\{\psi_\hbar\}$ and denote their Husimi distributions as $H_\hbar$. Consider the limit
\begin{equation}\label{EqSemiMH}
\lim_{\hbar\to0}\int_\Omega H_{\hbar}(q,p)a(q,p)\,dqdp.
\end{equation}
It turns out that this limit coincides with (\ref{EqSemiM}). The definition of the semiclassical measure given in \cite{Martinez} is based exactly on the limit  (\ref{EqSemiMH}) for the Husimi distributions. But only the case of $\mathbb R^d$ and the Gaussian coherent states are considered there. Let us prove the equivalence of definitions  (\ref{EqSemiM}) and (\ref{EqSemiMH}) for the case of torus and for arbitrary coherent states of form (\ref{EqCoher}), (\ref{EqCoherRd}). We need an additional property.

\begin{proposition}\label{PropSmear}
The Husimi distribution of an arbitrary function $\psi\in L^2(\mathbb T^d)$  can be expressed as
\begin{equation}\label{EqHusimiWigner}
H_\psi(q,p)=\int_{\mathbb R^{2d}}W^{(\hbar)}_{q,p}(q',p')W_\psi(q',p')\,dq'dp',
\end{equation}
where  $W^{(\hbar)}_{q,p}$ and $W_\psi$ denote the Wigner distributions of the coherent state $\upsilon^{(\hbar)}_{q,p}$ and of the function $\psi$, respectively.
\end{proposition}

This is a known relation between the Wigner and Husimi functions in $\mathbb R^d$ \cite{Hillery,McKenna}; let us prove if for $\mathbb T^d$. First of all, let us note that, in view of (\ref{EqWdelta}), the Wigner distribution of the coherent state $W^{(\hbar)}_{q,p}$ is a smooth function (regular distribution), so, expression  (\ref{EqHusimiWigner}) is well-defined (as the action of the generalized function $W_\psi$ on the test function $W^{(\hbar)}_{q,p}$).

\begin{proof}
We have
\begin{eqnarray*}\fl
H_\psi(q,p)=\frac1{(2\pi\hbar\alpha_\hbar)^d}
\sum_{n,m\in\mathbb Z^d}
\int_{\mathbb T^{2d}}dxdy\,\overline{\psi(x)}\psi(y)\varphi\left(\frac{x-2\pi n-q}{\alpha_\hbar}\right)
\overline{\varphi\left(\frac{y-2\pi m-q}{\alpha_\hbar}\right)}\\
\qquad\qquad\qquad\qquad\:\:\,\times\exp\left\lbrace\frac{ip[(x-2\pi n)-(y-2\pi m)]}\hbar\right\rbrace
\\=
\frac1{(2\pi\hbar\alpha_\hbar)^d}
\int_{\mathbb R^{2d}}dxdy\,
\overline{\psi(x)}\psi(y)\varphi\left(\frac{x-q}{\alpha_\hbar}\right)
\overline{\varphi\left(\frac{y-q}{\alpha_\hbar}\right)}\exp\left[\frac{ip(x-y)}\hbar\right]\\=
\frac1{(\pi\hbar\alpha_\hbar)^d}
\int_{\mathbb R^{2d}}dq'dx\,
\overline{\psi(q'+x)}\psi(q'-x)\varphi\left(\frac{q'-q+x}{\alpha_\hbar}\right)
\overline{\varphi\left(\frac{q'-q-x}{\alpha_\hbar}\right)}\\\qquad\qquad\quad\:\times\exp\left[\frac{2ipx}\hbar\right]\\=
\int_{\mathbb R^{2d}}dq'dp'
\int_{\mathbb R^d}\frac{dx}{(\pi\hbar)^d}
\overline{\psi(q'+x)}\psi(q'-x)\exp\left[\frac{2ip'x}\hbar\right]
\\\qquad\quad\times
\int_{\mathbb R^d}\frac{dy}{(\pi\hbar\alpha_\hbar)^d}
\varphi\left(\frac{q'-q-y}{\alpha_\hbar}\right)
\overline{\varphi\left(\frac{q'-q+y}{\alpha_\hbar}\right)}\exp\left[\frac{2i(p'-p)y}\hbar\right]
\\=\int_{\mathbb R^{2d}}W_\psi(q',p')W^{(\hbar)}_{q,p}(q',p')\,dq'dp',
\end{eqnarray*}
Q.E.D.\end{proof}

Due to (\ref{EqWquni})--(\ref{EqWpuni}), (\ref{EqHusimiWigner}) can be rewritten as
\begin{equation*}
H_\psi(q,p)=\int_{\mathbb R^{2d}}W^{(\hbar)}_{0,0}(q'-q,p'-p)W_\psi(q',p')\,dq'dp'+o(1),\quad \hbar\to0.
\end{equation*}
As a corollary of this formula and Proposition~\ref{PropDelta}, we have the following.
\begin{proposition}
The existence of limit (\ref{EqSemiM}) is equivalent to the existence of limit (\ref{EqSemiMH}), and both limits coincide.
\end{proposition}

\begin{remark}\label{RemHusimi}
The Husimi distribution has a direct physical meaning. Consider the probability operator-valued measure $M$ defined as
\begin{equation*}
M(B)=\frac1{(2\pi\hbar)^d}\int_\Omega P[\upsilon^{(\hbar)}_{qp}]\,dqdp,
\end{equation*}
where $B\subset\Omega$ is an arbitrary Borel set on the phase space. According to the formalism of quantum mechanics \cite{Holevo}, it can be interpreted as an observable corresponding to simultaneous approximate measurements of position and momentum. Observables of this type were introduced by von Neumann \cite{Neumann}. His motivation was as follows. In classical mechanics, simultaneous measurements of position and momentum is possible. Hence, the correspondence principle requires this to be possible approximately also in quantum mechanics (with the errors of measurements tending to zero as $\hbar\to0$). If we choose the Gaussian function $\varphi(x)$ (\ref{EqGauss}), then the product of errors of measurements of position and momentum is $\hbar/2$, which minimizes the uncertainty relations. We see that the Husimi distribution is nothing else but the distribution of outcomes of such measurements.
\end{remark}

\section{Main results}\label{SecMain}

Let us denote
\begin{equation*}
\upsilon^{(\hbar)}_{q_0p_0,t}=\exp(-i\hbar t\Delta)\upsilon^{(\hbar)}_{q_0p_0}
\end{equation*}
--- a wave packet evolved on the time $t\in\mathbb R$,
\begin{equation}
H_{q_0,p_0,t}(q,p)=\frac1{(2\pi\hbar)^d}|(\upsilon_{qp},\upsilon_{q_0p_0,t})|^2
\end{equation}
--- the corresponding Husimi distribution. Also recall that  $T_\hbar$ denotes the revival time (\ref{EqT}).

\begin{theorem}\label{ThMain}
Consider a real-valued function $t_\hbar$ of $\hbar$ such that $\hbar (t_\hbar-AT_\hbar)\to0$,\\ $\hbar (t_\hbar-AT_\hbar)/\alpha_\hbar\to B$, where $A\in\mathbb R$, $B\in[0,\infty]$. Another possibility is $\hbar t_\hbar\to+\infty$. Then:

1) If $B=\infty$, or $A$ is irrational, or $\hbar t_\hbar\to+\infty$, then
\begin{equation}\label{EqLimFlat}
\lim_{\hbar\to0}H_{q_0,p_0,t_\hbar}(q,p)=
\frac1{(2\pi)^d}\,\delta(p-p_0);
\end{equation}

2) If $B<\infty$ and $A=\frac MN$ is rational (expressed as an irreducible fraction), then
\begin{eqnarray}\fl
\lim_{\hbar\to0}
\lbrace
H_{q_0,p_0,t_\hbar}(q,p)\nonumber\\-
\frac1{N'}\sum_{l\in[N']^d}\delta_B\left(q-q_0-2p_0(t_\hbar-AT_\hbar)-\Delta q_0+\frac{2\pi k}{N'}\right)\delta(p-p_0)\label{EqLim}
\rbrace=0.
\end{eqnarray}
Here
\begin{equation*}
\delta_B(q)=\frac1{(2\pi)^d}\sum_{j\in\mathbb Z^d}\sigma_{Bj}\exp(ijq),\\
\end{equation*}
\begin{equation*}
\sigma_{R}=\int_{\mathbb R^d}\varphi(x)\overline{\varphi(x-2R)}\,dx,
\end{equation*}
$[N']=\{0,1,2,\ldots,N'-1\}$,
\begin{equation*}
N'=\cases{
N,&odd $N$,\\
\frac N2,&even $N$,}
\qquad\qquad
\Delta q_0=\cases{\frac{2\pi}{N} I,&$N\equiv2\pmod 4$,\\
0,&otherwise,}
\end{equation*}
where $I=(1,1,\ldots,1)\in\mathbb Z^d$.
\end{theorem}

\begin{theorem}\label{ThMu}
Semiclassical measures $\mu$ corresponding to time-evolved wave packets $\upsilon^{(\hbar)}_{qp,t_\hbar}$ are:

1) Let $t_\hbar$ is as in Theorem~\ref{ThMain}; $B=\infty$, or $A$ is irrational, or $\hbar t_\hbar\to+\infty$. Then

\begin{equation}\label{EqMuFlat}
\mu(dqdp)=\frac1{(2\pi)^d}\,\delta(p-p_0)\,dqdp;
\end{equation}

2) Let $t_\hbar$ is as in Theorem~\ref{ThMain}; $B<\infty$, $A=\frac MN$, $p_0=0$. Then
\begin{equation*}
\mu(dqdp)=
\frac1{N'}\sum_{k\in[N']^d}
\delta_B\left(q-q_0-\Delta q_0+\frac{2\pi k}{N'}\right)\delta(p)\,dqdp;
\end{equation*}

3) Let $t_\hbar$ be a real-valued function of $\hbar$ such that $t_\hbar-\frac MNT_\hbar\to\tau\in\mathbb R$. Then

\begin{equation*}
\mu(dqdp)=
\frac1{N'}\sum_{k\in[N']^d}
\delta\left(q-q_0-2p_0\tau-\Delta q_0+\frac{2\pi k}{N'}\right)\delta(p-p_0)\,dqdp.
\end{equation*}
\end{theorem}

\begin{theorem}\label{ThTime}
Consider the function $t_\hbar=\lambda_\hbar t$, where $\lambda_\hbar\to\infty$ as $\hbar\to0$. 

1) If $\hbar\lambda_\hbar/\alpha_\hbar\to\infty$ as $\hbar\to0$, then, for all functions $a\in C^\infty_0(\Omega)$ and $b\in L^1(\mathbb R)$, there exists the limit
\begin{equation}\label{EqLimt}
\lim_{\hbar\to0}\int_{\Omega\times\mathbb R}a(q,p)b(t)H_{q_0,p_0,\lambda_\hbar t}(q,p)\,dqdpdt=\langle a\rangle(p_0)\int_{\mathbb R}b(t)\,dt,
\end{equation}
where 
\begin{equation*}
\langle a\rangle(p_0)=\frac1{(2\pi)^d}\int_{\mathbb T^d}a(q,p_0)\,dq.
\end{equation*}

2) If $\hbar\lambda_\hbar/\alpha_\hbar\to B\in[0,+\infty)$, then, for all functions $a\in C^\infty_0(\Omega)$ and $b\in L^1(\mathbb R)$, there exists the limit
\begin{equation}\label{EqLimt2}\fl
\lim_{\hbar\to0}\int_{\Omega\times\mathbb R}a(q,p)b(t)H_{q_0,p_0,\lambda_\hbar t}(q,p)\,dqdpdt=\lim_{T\to\infty}\frac1T\int_0^T a^{(b,B)}(q+p_0t,p_0)\,dt,
\end{equation}
where
\begin{equation*}\fl
a^{(b,B)}(q,p)=\frac1{(2\pi)^{\frac d2}}\sum_{j\in\mathbb Z^d}a_j(p)\left[\int_{\mathbb R^d}dx\int_{-\infty}^{+\infty}b(t)\varphi(x)\overline{\varphi(x-2Btj)}\,dt\right]\exp(ijq).
\end{equation*}
In particular, if $B=0$, then
\begin{equation}\label{EqLimt3}\fl
\lim_{\hbar\to0}\int_{\Omega\times\mathbb R}a(q,p)b(t)H_{q_0,p_0,\lambda_\hbar t}(q,p)\,dqdpdt=\lim_{T\to\infty}\frac1T\int_0^T a(q+p_0t,p_0)dt
\int_{-\infty}^{+\infty}b(t)dt.
\end{equation}
If $p_0$ does not belong to the ``resonant'' set
\begin{equation}\label{EqReson}
R=\{p\in\mathbb R^d\,|\,jp=0  \textrm{ for some }  j\in\mathbb Z^d\backslash\{0\}\},
\end{equation}
then limits (\ref{EqLimt2}) and (\ref{EqLimt3}) are reduced to (\ref{EqLimt}).
\end{theorem}

\begin{proof}[Proof of Theorem~\ref{ThMain}]
Let us find the Fourier coefficients of coherent states
 $\upsilon^{(\hbar)}_{q_0p_0}(x)$:
\begin{equation}
\upsilon^{(\hbar)}_{q_0p_0}(x)=
\frac1{{(2\pi)}^{\frac d2}}\sum_{k\in\mathbb Z^d}c^{(\hbar)}_{k,qp}\exp(ikx).
\end{equation} 
We have
\begin{eqnarray*}
c^{(\hbar)}_{k,qp}&=\frac1{(2\pi)^{\frac d2}}\int_{\mathbb T^d}\upsilon^{(\hbar)}_{qp}(x)\exp(-ikx)\,dx\\&=
\frac1{(2\pi)^{\frac d2}}\sum_{n\in\mathbb Z^d}\int_{\mathbb T^d}\eta^{(\hbar)}_{qp}(x-2\pi n)\exp[-ik(x-2\pi n)]\,dx\\&=
\frac1{(2\pi)^{\frac d2}}\int_{\mathbb R^d}\eta^{(\hbar)}_{qp}(x)\exp(-ikx)\,dx\\&=
\left(\frac{\alpha_\hbar}{2\pi}\right)^{\frac d2}\int_{\mathbb R^d}\varphi(x)
\exp\left\lbrace i(\alpha_\hbar x+q)\left(\frac p\hbar-k\right)-\frac{ipq}\hbar\right\rbrace dx.
\end{eqnarray*}
Using this formula, we can calculate the scalar product
\begin{equation*}
(\upsilon_{qp},\upsilon_{q_0p_0,t})=
\sum_{k\in\mathbb Z^d}
\overline{c^{(\hbar)}_{k,qp}}c^{(\hbar)}_{k,q_0p_0}\exp(-i\hbar tk^2).
\end{equation*}
Let a function $a(q,p)\in C^\infty_0(\Omega)$ is expanded into the Fourier series and the Fourier integral as follows:
\begin{equation*}\fl
a(q,p)=\frac1{(2\pi)^{\frac d2}}
\sum_{j\in\mathbb Z^d}a_j(p)\exp(ijq)=
\frac1{(2\pi)^d}\sum_{j\in\mathbb Z^d}\int_{\mathbb R^d}
 \tilde a_j(\xi)\exp(ijq+i\xi p)d\xi.
\end{equation*}
Then
\begin{eqnarray}\label{EqAFourier}\fl
\int_\Omega H_{q_0,p_0,t_\hbar}^{(\hbar)}(q,p)a(q,p)\,dqdp\\=
\frac1{(2\pi\sqrt\hbar)^{2d}}\sum_{j\in\mathbb Z^d}\int_{\mathbb R^d}d\xi \tilde a_j(\xi)
\int_\Omega 
|(\upsilon_{qp},\upsilon_{q_0p_0,t_\hbar})|^2\exp(ijq+i\xi p)\,dqdp.
\end{eqnarray}
Let us calculate 
\begin{eqnarray*}
\frac1{(2\pi\hbar)^d}\int_\Omega 
|(\upsilon_{qp},\upsilon_{q_0p_0,t_\hbar})|^2\exp(ijq+i\xi p)dqdp=\\=
\frac1{(2\pi\hbar)^d}\left(\frac{\alpha_\hbar}{2\pi}\right)^{2d}
\sum_{k,n\in\mathbb Z^d}
\int_{\mathbb R^{4d}}dxdx'dydy'\int_\Omega dqdp\,
\overline{\varphi(x)}\varphi(y)\varphi(x')\overline{\varphi(y')}
\quad\\\times
\exp\left\lbrace
-i(\alpha_\hbar x+q)\left(\frac p\hbar-k\right)
+i(\alpha_\hbar y+q_0)\left(\frac {p_0}\hbar-k\right)-iht_\hbar k^2
+\right.\\\left.
\quad+i(\alpha_\hbar x'+q)\left(\frac p\hbar-n\right)
-i(\alpha_\hbar y'+q_0)\left(\frac {p_0}\hbar-n\right)+iht_\hbar n^2
+ijq+i\xi p
\right\rbrace.
\end{eqnarray*}
The integration over $p$ yields the factor
\begin{equation*}
(2\pi)^d\delta\left(\xi-\frac{\alpha_\hbar}\hbar(x-x')\right)=
\left(\frac{2\pi\hbar}{\alpha_\hbar}\right)^d
\delta\left(x'-x+\frac{\hbar\xi}{\alpha_\hbar}\right).
\end{equation*}
The integration over $q$ yields the factor
$(2\pi)^d\delta_{j+k-n}$, where $\delta_x$ is the Kronecker symbol ($\delta_x=1$ if $x=0$ and $\delta_x=0$ otherwise). Thus, the integration over $x'$ and the summation over $n$ can be eliminated with the substitutions 
$x'=x-\frac{\hbar\xi}{\alpha_\hbar}$ and $n=k+j$. We have
\begin{eqnarray*}\fl
\frac1{(2\pi\hbar)^d}\int_\Omega 
|(\upsilon_{qp},\upsilon_{q_0p_0,t_\hbar})|^2\exp(ijq+i\xi p)dqdp\\=
\left(\frac{\alpha_\hbar}{2\pi}\right)^{d}
\sum_{k\in\mathbb Z^d}
\int_{\mathbb R^{3d}}dxdydy'\,
\overline{\varphi(x)}\varphi(y)\varphi\left(x-\frac{\hbar\xi}{\alpha_\hbar}\right)\overline{\varphi(y')}
\\\times
\exp\Big\lbrace
\frac{i\alpha_\hbar}{\hbar}p_0(y-y')+i k[\alpha_\hbar(y'-y)+2\hbar t_\hbar j+\hbar\xi] - i\alpha_\hbar(x-y')j\\
\qquad\:\:+ij(q_0+i\hbar t_\hbar j)+i\hbar j\xi
\Big\rbrace.
\end{eqnarray*}
Here we can drop the infinitesimal terms $i\hbar j\xi$ and $i\alpha_\hbar(x-y')j$ in the exponent. Also note that
\begin{equation*}
\lim_{\hbar\to0}\int_{\mathbb R^d}\overline{\varphi(x)}\varphi\left(x-\frac{\hbar\xi}{\alpha_\hbar}\right)dx=\int_{\mathbb R^d}|\varphi(x)|^2dx=1.
\end{equation*}
 Further,
the summation over $k$ yields the product
\begin{eqnarray*}
&(2\pi)^d\sum_{k\in\mathbb Z^d}
\delta(\alpha_\hbar(y'-y)+2\hbar t_\hbar j+2\pi k+\hbar\xi)\\=
&\left(\frac{2\pi}{\alpha_\hbar}\right)^d\sum_{k\in\mathbb Z^d}\delta\left(y'-y+\frac{2\hbar t_\hbar j+2\pi k+\hbar\xi}{\alpha_\hbar}\right).
\end{eqnarray*}
By the elimination of the integration over $y'$ with the substitution $y'=y-(2\hbar t_\hbar j+2\pi k+\hbar\xi)/\alpha_\hbar$, we obtain
\begin{eqnarray*}\fl
\lim_{\hbar\to0}\Big[\frac1{(2\pi\hbar)^d}\int_\Omega 
|(\upsilon_{qp},\upsilon_{q_0p_0,t_\hbar})|^2\exp(ijq+i\xi p)dqdp-\\
\sum_{k\in\mathbb Z^d}
\int_{\mathbb R^d}dy\,\varphi(y)\overline{\varphi\left(y-\frac{2\hbar t_\hbar j+2\pi k+\hbar\xi}{\alpha_\hbar}\right)}\\
\qquad\quad\:\:\times
\exp\left\lbrace
ij(q_0+2p_0t_\hbar)+i\xi p_0+2\pi k\frac{p_0}\hbar+i\hbar t_\hbar j^2
\right\rbrace\Big]=0.
\end{eqnarray*}
If we drop the infinitesimal term $\hbar\xi/\alpha_\hbar$ in the argument of the function $\overline\varphi$ and substitute the result into formula (\ref{EqAFourier}), we will obtain
\begin{eqnarray}\fl
\lim_{\hbar\to0}\Big[
\int_\Omega H^{(\hbar)}_{q_0,p_0,t_\hbar}(q,p)a(q,p)\,dqdp\nonumber\\
-\frac1{(2\pi)^{\frac d2}}\sum_{j,k\in\mathbb Z^d}a_j(p_0)\int_{\mathbb R^d}dx\,\varphi(x)\overline{\varphi\left(x-\frac{2\hbar t_\hbar j+2\pi k}{\alpha_\hbar}\right)}\nonumber\\\qquad\qquad\qquad\times
\exp\left\lbrace
ij(q_0+2p_0t_\hbar)+2\pi k\frac{p_0}\hbar+i\hbar t_\hbar j^2
\right\rbrace\Big]=0.\label{EqFinal}
\end{eqnarray}

Now consider all limiting cases. At first, as we see, we have the Dirac measure for the momentum, which was expected since the momentum conservation.

If $A$ is a whole number and $B<\infty$, then, in the summation over $k$ in (\ref{EqFinal}), only the term with $k=-2A$ holds (otherwise the integral over $x$ tends to zero due to rapid decrease of $\varphi$). We can see that we obtain formula (\ref{EqLim}) for the corresponding case. 

If $B=\infty$, or $A$ is irrational, or $\hbar t_\hbar\to\infty$, then, in the double sum in (\ref{EqFinal}), only the term with $j=k=0$ remains non-zero in the limit. This corresponds to the uniform spatial distribution. So, we obtain formula (\ref{EqLimFlat}).

Let now $A=\frac MN$ (rational number expressed as an irreducible fraction) and $B<\infty$. The integral in (\ref{EqFinal}) does not tend to zero if and only if  $\frac{2Mj}N+k=0$. Accordingly, in the summation over $j$, only terms with $j=N'\ell$, $\ell\in\mathbb Z^d$ remains non-zero in the limit, where $N'=N$ for odd $N$ and $N'=\frac N2$ for even $N$. In the summation over $k$, only the term with $k=-\frac{2Mj}N$ remains non-zero in the limit.

Consider the term $i\hbar t_\hbar j^2$ in the exponent in the right-hand side of (\ref{EqFinal}). If $N$ is odd then 
\begin{equation*}
\hbar t_\hbar j^2\sim2\pi\frac MN(N\ell)^2=2\pi MN\ell^2\in2\pi\mathbb Z^d
\end{equation*}
(we write $f\sim g$ whenever $\lim\frac fg=1$) and this term may be dropped. If $N$ is even, then 
\begin{equation*}
\hbar t_\hbar j^2\sim2\pi \frac MN\left(\frac{N\ell}2\right)^2=\pi\frac{MN\ell^2}2.
\end{equation*}
If $N$ is divisible by four, then this number again belongs to $2\pi\mathbb Z^d$ and may be dropped. If $N$ is even, but not divisible by four, then $M$ is odd and 
\begin{equation*}
\exp\left\lbrace i\hbar t_\hbar j^2 \right\rbrace \sim\exp\left\lbrace i\pi\frac{MN\ell^2}2\right\rbrace=(-1)^{N'\ell I}=
\exp\left\lbrace i\pi N'\ell I\right\rbrace.
\end{equation*}
Thus, the second term in the limiting expression in (\ref{EqFinal}) can be rewritten as
\begin{eqnarray}\label{EqSums}
\frac1{(2\pi)^{\frac d2}}&\sum_{\ell\in\mathbb Z^d}\sigma_{BN'\ell}a_{N'\ell}(p_0)
\exp\left\lbrace
iN'\ell[q_0+2p_0(t_\hbar-AT_\hbar)]+\gamma i\pi N'\ell I
\right\rbrace\nonumber\\=
\frac1{N'}&\sum_{k\in[N']^d}a^{(B)}\left(q_0+2p_0(t_\hbar-AT_\hbar)+\gamma\pi I+\frac{2\pi k}{N'},p_0\right)\nonumber\\=
\frac1{N'}&\sum_{k\in[N']^d}a^{(B)}\left(q_0+2p_0(t_\hbar-AT_\hbar)+\Delta q_0+\frac{2\pi k}{N'},p_0\right),
\end{eqnarray}
where $\gamma=1$ if $N\equiv2\pmod 4$ and $\gamma=0$ otherwise; 
\begin{equation}\label{Eqab}
a^{(B)}(q,p)=\frac1{(2\pi)^{\frac d2}}\sum_{j\in\mathbb Z^d}\sigma_{Bj}a_j(p)\exp(ijq).
\end{equation}
To verify the first equality in (\ref{EqSums}), we can use formula (\ref{Eqab}) and see that all terms except $j=N'\ell$ cancel. Replacement of $\pi$ by $\frac{2\pi}N$ in the second equality in (\ref{EqSums}) (recall that $\Delta q_0=\gamma\frac{2\pi}NI$) is valid since $k=(N'+1)/2$,
\begin{equation*}
\pi+\frac{2\pi k}{N'}=\pi+\frac{2\pi}{N'}\frac{N'+1}2=2\pi+\frac{2\pi}N,
\end{equation*}
and $(2\pi\mathbb Z^d)$-periodicity of $a^{(B)}(q,p)$ with respect to $q$. Finally, we obtain
\begin{eqnarray*}\fl
\lim_{\hbar\to0}\bigg[
\int_\Omega H^{(\hbar)}_{q_0,p_0,t_\hbar}(q,p)a(q,p)\,dqdp\\-\frac1{N'}\sum_{k\in[N']^d}a^{(B)}\left(q_0+2p_0(t_\hbar-AT_\hbar)+\Delta q_0+\frac{2\pi k}{N'},p_0\right)\bigg],
\end{eqnarray*}
i.e., formula (\ref{EqLim}). Thus, the theorem has been entirely proved.
\end{proof}

Theorem~\ref{ThMu} is a direct corollary of Theorem~\ref{ThMain}.

\begin{proof}[Proof of Theorem~\ref{ThTime}]
Consider the first case. According to Theorem~\ref{ThMain}, 
\begin{equation*}
\lim_{\hbar\to0}\int_\Omega H_{q_0,p_0,\lambda_\hbar t}^{(\hbar)}(q,p)a(q,p)dqdp=
\langle a\rangle(p_0)
\end{equation*}
for all $t$ (if $\hbar \lambda_\hbar\to0$ or $\hbar \lambda_\hbar\to\infty$) or for irrational $t$ (if $\hbar \lambda_\hbar\to c\in(0,\infty)$). Since rational numbers have zero measure on the real line, anyway,
\begin{equation}\label{EqThTime2}
\lim_{\hbar\to0}\int_{-\infty}^{+\infty}dt\,b(t)\int_\Omega H_{q_0,p_0,\lambda_\hbar t}^{(\hbar)}(q,p)a(q,p)dqdp=
\langle a\rangle(p_0)\int_{-\infty}^{+\infty}b(t)\,dt.
\end{equation}

Consider the second case. Let us rewrite formula (\ref{EqFinal}) for this case (recall that the terms with $k\neq0$ vanish in this limiting case):
\begin{eqnarray}
\lim_{\hbar\to0}\Big[
\int_{\Omega\times\mathbb R}H^{(\hbar)}_{q_0,p_0,t_\hbar}(q,p)a(q,p)b(t)\,dqdpdt\nonumber\\\qquad-\frac1{(2\pi)^{\frac d2}}\sum_{j\in\mathbb Z^d}a_j(p_0)\int_{\mathbb R^d}dx\int_{-\infty}^{+\infty}dt\,b(t)\varphi(x)\overline{\varphi\left(x-\frac{2\hbar \lambda_\hbar tj}{\alpha_\hbar}\right)}\nonumber\\\qquad\qquad\qquad\times
\exp\left\lbrace
ij(q_0+2p_0\lambda_\hbar t)\right\rbrace\Big]=0.\label{EqThTime3}
\end{eqnarray}
By the Riemann--Lebesgue theorem, the integral over $t$ tends to zero as $jp_0\neq0$ due to the term $2ijp_0\lambda_\hbar t$ in the exponent. Hence,
\begin{equation*}\fl
\lim_{\hbar\to0}\Big[
\int_{\Omega\times\mathbb R}H^{(\hbar)}_{q_0,p_0,t_\hbar}(q,p)a(q,p)b(t)\,dqdpdt-\frac1{(2\pi)^{\frac d2}}\sum_{j:\,jp_0=0}a^{(b,B)}_j(p_0)
\exp(ijq_0)\Big]=0,
\end{equation*}
which can be rewritten as (\ref{EqLimt2}). If $B=0$, then  (\ref{EqLimt2}) can be obviously rewritten as  (\ref{EqLimt3}).

If $p_0$ does not belong to the resonant set, then, in (\ref{EqThTime3}), only the term with $j=0$ remains non-zero in the limit. Since
\begin{equation*}
a_0(p)=\frac1{(2\pi)^{\frac d2}}\int_{\mathbb R^d}a(q,p)\,dq,
\end{equation*}
(\ref{EqLimt2}) and (\ref{EqLimt3}) can be rewritten as (\ref{EqLimt}).
\end{proof}

\section{Discussion}\label{SecDiscus}

\subsection{Three time scales}

From theorem~\ref{ThMain} three time scales can be deduced:

\begin{enumerate}
\item ``Classical'' time scale. If $t_\hbar=t=const$, or $t_\hbar\to\infty$ but $\hbar t_\hbar/\alpha_\hbar\to0$, then the wave packet moves along the classical trajectory: the second term in the limit (\ref{EqLim}) has the form
\begin{equation*}
\delta(q-q_0-2p_0t_\hbar)\delta(p-p_0).
\end{equation*}
Conventionally, as a characteristic duration of this time scale can be chosen as the classical ``period'' of motion $T_{cl}=\pi/{\overline p}$, where $\overline p=\frac1d\sum_{j=1}^d p_j$ is the mean momentum for  $p=(p_1,\ldots,p_d)\in\mathbb R^d$.

\item $T_{coll}=\alpha_\hbar/\hbar$ is a characteristic time of the collapse of the wave packet. The rate of wave packet spreading is known to be proportional to the initial standard deviation of the momentum. The standard deviation of the Gaussian wave packet is equal to $\hbar/(2\alpha_\hbar)$;

\item $T_\hbar=2\pi/\hbar$ is the full revival time. The instants $\frac MNT_\hbar$ correspond to fractional revivals. They correspond to revivals of small copies of the wave packet in several points on the torus. The structure of fractional revivals for the general case of systems with discrete spectrum was elaborated in \cite{Averbuh,Averbuh2}. A more detailed analysis for the infinite square well is given in \cite{AronStroud,AronStroud00} (the motion in the infinite square well is equivalent to the free motion on the torus \cite{VolTrush}). For a further development of the general theory of fractional revivals see \cite{Schleich-prl,Schleich-pra,Robinett00,Robinett,AronStroud05}.
\end{enumerate}

The Ehrenfest time is $O(T_{coll})=O(\alpha_\hbar/\hbar)$. By proper choices of $\alpha_\hbar$ we can made the Ehrenfest time arbitrarily close from below to  $O(\hbar^{-1})$ and can made it arbitrarily small (but still indefinitely increasing as $\hbar\to0$), i.e., even smaller than $O(\ln\hbar^{-1})$.

\subsection{Rational and irrational times}

Theorems~\ref{ThMain} and~\ref{ThMu} distinguish rational and irrational $A$. However, every irrational $A$ can be approximated by rationals $M/N$, where $M\to\infty$ and $N\to\infty$ such that $M/N\to A$. Hence, rational (with large denominators $N$) and irrational $A$ should be physically indistinguishable. This is true in our case as well: if $N\to\infty$, then, according to (\ref{EqLim}), the number of small copies of the wave packet tends to infinity and their centres are uniformly distributed on the torus. Thus, the spatial distribution produced by the sum of many tiny wave packets tends to the uniform distribution. So, the cases of rational $A=M/N$ with large $N$ and irrational $A$ are indeed physically indistinguishable if a measurement instrument has a finite precision.

The distinguish of rational and irrational times (in the units of $T_\hbar$) in the semiclassical limit reveals the relation of quantum mechanics to number-theoretic issues discovered in some other models \cite{Kara,NumFactoring,Morse}.

\subsection{Generalizations of semiclassical measures}

To formulate the results in terms of semiclassical measures in Theorem~\ref{ThMu}, we had to narrow the class of functions $t_\hbar$ (in comparison to Theorem~\ref{ThMain}). This is due to the term $2p_0(t_\hbar-T_\hbar)$ in the argument of $\delta_B$ in (\ref{EqLim}). Generally, this term itself has no limit. The cases considered in Theorem~\ref{ThMu} are related to different cases when this divergence is eliminated. This is possible either in the case of the uniform spatial distribution, when $\delta_B$ does not depend on the spatial arguments at all, or whenever $p_0=0$, or whenever $t_\hbar-T_\hbar$ converges to a constant.

Another way of obtaining the convergent expressions for semiclassical measures is time-averaging. This way was used in \cite{MaciaRiemann, MaciaTorus, AnaMaciaTorus, AnaMaciaView, Ana14}. We consider it in Theorem~\ref{ThTime}.  Let us reformulate this theorem from a more general viewpoint developed in the aforementioned works.

Let $\{\psi_\hbar\}$ is a family of functions; $t_\hbar=\lambda_\hbar t$, where $t\in\mathbb R$, and $\lambda_\hbar\to\infty$ as $\hbar\to0$. Denote  by
$W_\hbar(q,p,t)$ the Wigner distribution of the function $\exp(-i\lambda_\hbar t\Delta)\psi_\hbar$.
Then, if for all functions $a\in C^\infty_0(\Omega)$ and $b\in L^1(\mathbb R)$ there exists the limit
\begin{equation}\label{EqSemiMt}
\lim_{\hbar\to0}\int_{\Omega\times\mathbb R}a(q,p)b(t)W(q,p,t)\,dqdpdt=
\int_{\Omega\times\mathbb R}a(q,p)b(t)\mu_t(dqdp)dt,
\end{equation}
where the time-dependent measure $\mu_t(\Omega)$ is finite and bounded as a function of  $t$, then $\mu_t$ is also called the (time-dependent) semiclassical measure. According to Theorem~\ref{ThTime}, for coherent states, if $\hbar\lambda_\hbar/\alpha_\hbar\to\infty$ or $p_0$ does not belong to the set of resonant frequencies, we have
\begin{equation*}
\mu_t(dqdp)=\frac1{(2\pi)^d}\delta(p-p_0)
\end{equation*}
for all $t$, i.e., uniform spatial distribution.

However, as we see, this approach does not distinguish all three time scales. If $\hbar\lambda_\hbar/\alpha_\hbar\to0$, then we have the classical time scale; if $\hbar\lambda_\hbar/\alpha_\hbar\to B\in(0,\infty)$, or $\hbar\lambda_\hbar/\alpha_\hbar\to \infty$ but $\hbar\lambda_\hbar\to0$, then we have  the collapse time scale; if $\hbar\lambda_\hbar\to A>0$, then we have the revival time scale.

All three cases gives the uniform spatial distribution in the case of time-averaging, but the  reasons are different. If $\hbar\lambda_\hbar/\alpha_\hbar\to B\in[0,\infty)$, then the cause of the uniformity of the spatial distribution is the averaging over the classical trajectory (if the mean momentum does not belong to resonant frequencies, this is a necessary condition in this case). If $\hbar\lambda_\hbar/\alpha_\hbar\to \infty$, then the cause of the uniformity is not the averaging over the classical trajectory, but the actual collapse of the wave packet (and the uniformity takes place irrespectively of whether the mean momentum belongs to resonant frequencies). Moreover, such important and interesting wave phenomenon like wave packet revivals is fully missed in the approach based on time-averaging.

This demonstrates limitations of the approach to long-time quantum dynamics based on semiclassical measures. Other limitations were reviewed in \cite{Carles}. As an alternative, one can consider the approach of Theorem~\ref{ThMain}, where, instead of limits of the Husimi distributions themselves, distributions equivalent to the Husimi distributions in the corresponding long-time semiclassical limits are under consideration.

Also we can try to modify the definition of the semiclassical measure by introduction a correction for the classical phase flow. Let us denote $g^t(q,p)$ the displacement of the point  $(q,p)$ on $t$ along the classical phase trajectory. In the case of free motion on the torus, $g^t(q,p)=(q+2pt,p)$. If $\hbar(t_\hbar-T_\hbar)\to0$, then define

\begin{equation}\label{EqSemiMf}
\lim_{\hbar\to0}\int_\Omega W_{\hbar}(g^{t_\hbar+AT_\hbar}(q,p))a(q,p)\,dqdp=\int_\Omega a(q,p)\omega(dqdp).
\end{equation}
Then formula (\ref{EqLim}) of Theorem~\ref{ThMain} takes the form

\begin{equation*}
\omega(dqdp)=
\frac1{N'}\sum_{k\in[N']^d}
\delta_B\left(q-q_0-\Delta q_0+\frac{2\pi k}{N'}\right)\delta(p-p_0)\,dqdp.
\end{equation*}

An interesting question is a possibility of generalization of this approach. One of the difficulties is that the exact revival after some time $T_\hbar$ is a property of only quadratic Hamiltonians. In general, the dynamics of systems with discrete spectrum is not periodic, but almost periodic.

\subsection{Gaussian coherent states}
Let $\varphi$ be Gaussian (\ref{EqGauss}). Then
\begin{equation*}
\delta_B(q)=\frac1{(2\pi)^d}\sum_{j\in\mathbb Z^d}\exp\left\lbrace-\frac{(Bj)^2}2+ijq\right\rbrace=
\theta\left(\frac q{2\pi},\frac{B^2}{2\pi}\right),
\end{equation*}
where
\begin{equation*}
\theta(x,\tau)=\sum_{k\in\mathbb Z^d}\exp\{-\pi\tau k^2+2\pi ikx\}
\end{equation*}
is the theta function of several variables, $x\in\mathbb R^d$, $\tau\in\mathbb C$, $\mathrm{Re}\,\tau>0$. Using the functional equation for the theta function \cite{Mum} 
\begin{equation*}
\theta\left(\frac{x}{i\tau},\frac{1}{\tau}\right)=\tau^{\frac d2} \exp\left(\frac{\pi x^2}\tau\right)\theta(x,\tau),
\end{equation*}
we arrive at
\begin{equation*}
\delta_B(q)=\frac1{(2\pi B^2)^\frac d2}\sum_{n\in\mathbb Z^d}
\exp\left\lbrace-\frac{(q-2\pi n)^2}{2B^2}\right\rbrace.
\end{equation*}
So, in this case, $B$ is the spatial standard deviation of the wave packet. We have reproduced the corresponding results of work \cite{VolTrush}.

\subsection{Physically small parameters}\label{SecPhysSmall}

We mentioned in Sec.~\ref{SecPrelim} that, physically, the Planck constant cannot tend to zero and one should speak about the smallness of certain dimensionless quantities. In our case the condition $\hbar\to0$, $\alpha_\hbar\to0$, $\hbar/\alpha_\hbar\to0$ is equivalent to the condition that every time scale is much greater than the previous one, i.e., they are ``well distinguishable'':

\begin{equation*}
T_{rev}\gg T_{coll}\gg T_{cl}.
\end{equation*}

In other words we can say:
\begin{itemize}
\item $\alpha_\hbar\ll 2\pi $ means that the spatial extension of the wave packet is much smaller than the size of the torus (this corresponds to $T_{coll} \ll T_{rev}$);

\item $\overline p\pi \gg \hbar$ means that the physical action related to a single revolution of the particle around the torus is much greater than the quantum of action (this corresponds to $T_{rev}\gg T_{cl}$);

\item $\overline p\alpha_\hbar \gg \hbar$ means that the action related to the motion of the center of the wave packet  along its spatial extension is much larger than the quantum of action (this is a strengthening of the previous condition; corresponds to $T_{coll} \gg T_{cl}$).
\end{itemize}

\section{Conclusions}

We have obtained explicit expressions for semiclassical measures corresponding to all stages of evolution of quantum wave packets on the flat torus: classical-like motion, spreading and revivals of the wave packet. The second time scale is the Ehrenfest time scale and the third one is beyond it. These explicit expressions allows to understand the limitations of the notion of semiclassical measure and to propose some generalizations.

The results can be applied to the particle in the infinite square well because it is reduced to the dynamics on the flat torus \cite{VolTrush}. An interesting problem would be the calculation of semiclassical measures for more general potentials, for example, the Morse potential, as well as various multi-dimensional bounded domains and compact manifolds. Coherent states for the Morse potential were constructed in \cite{Angelova}, the structure of revivals was studied in \cite{WangHeller,Morse}. One can also consider quantum optimal control problems (with both coherent and incoherent controls \cite{PechIlyn,PechTrush}) in semiclassical long-time  limit.

Our method of research was  direct summation of series of eigenvectors for time-evolved wavepackets, instead of reduction of the quantum dynamics to the classical dynamics usually applied in the semiclassical analysis. Though this method has already showed its effectiveness in the Jaynes--Cummings model \cite{Kara}, its possibilities for analysis of quantum mechanical models are still underexplored.

\section*{Acknowledgements}

Tha author is very grateful to M.\,V. Berry, S.\,Yu.~Dobrokhotov, A.\,S.~Holevo, J.\,R.~Klauder, I.\,V.~Volovich, and E.\,I.~Zelenov for helpful suggestions, comments, and interest to the work.

The work was supported by the Russian Science Foundation under grant 14-50-00005.

\end{document}